\documentclass[letterpaper, 10 pt, conference]{ieeeconf}
\IEEEoverridecommandlockouts
\usepackage{times}
\usepackage[english]{babel}
\usepackage{multicol}
\usepackage{hyperref}
\usepackage{overpic}
\usepackage{color}
\usepackage{soul}
\usepackage{bm}
\usepackage{graphicx}
\usepackage{float}
\usepackage{subcaption}
\usepackage{url}
\usepackage{cite} % make multiple references look nicer
\usepackage{amsmath}
\usepackage{amsfonts}
\usepackage{amssymb}
\usepackage{amsthm}
\usepackage{bm}
\usepackage{bbm}
\usepackage{mathtools}
\usepackage{thmtools,thm-restate}
\usepackage{algorithm}
\usepackage{algorithmic}
\usepackage{color}
\usepackage{graphicx}
\usepackage{comment}
\def\CC{\mathcal{C}}

\def\KK{\mathcal{K}}
\def\MM{\mathcal{M}}\def\NN{\mathcal{N}}

\def\TT{\mathcal{T}}

\def\Bb{\mathbf{B}}

\def\Gb{\mathbf{G}}\def\Ib{\mathbf{I}}

\def\Mb{\mathbf{M}}
\def\Pb{\mathbf{P}}

\def\ab{\mathbf{a}}
\def\fb{\mathbf{f}}
\def\gb{\mathbf{g}}

\def\ub{\mathbf{u}}
\def\vb{\mathbf{v}}\def\xb{\mathbf{x}}

\def\Rbb{\mathbb{R}}

\def\ett{\mathtt{e}}

\def\ltt{\mathtt{l}}

\def\rtt{\mathtt{r}}
\def\utt{\mathtt{u}}
\def\vtt{\mathtt{v}}

\def\R{\Rbb}

\def\t{\top}
\def\*{\star}
\def\zerob{{\mathbf0}}

\DeclareMathOperator*{\argmin}{arg\,min}

\usepackage[dvipsnames]{xcolor}

\newcommand{\q}{\mathbf{q}}
\newcommand{\qd}{{\dot{\q}}}
\newcommand{\qdd}{{\ddot{\q}}}

\newcommand{\x}{\mathbf{x}}
\newcommand{\xd}{{\dot{\x}}}
\newcommand{\xdd}{{\ddot{\x}}}
\newcommand{\y}{\mathbf{y}}
\newcommand{\yd}{{\dot{\y}}}
\newcommand{\ydd}{{\ddot{\y}}}
\newcommand{\z}{\mathbf{z}}
\newcommand{\zd}{{\dot{\z}}}
\newcommand{\zdd}{{\ddot{\z}}}

\newcommand{\f}{\mathbf{f}}

\newcommand{\J}{\mathbf{J}}
\newcommand{\Jd}{{\dot{\J}}}

\newcommand{\B}{\mathbf{B}}

\newcommand{\G}{\mathbf{G}}

\newcommand{\I}{\mathbf{I}}

\newcommand{\M}{\mathbf{M}}

\newcommand{\sdot}[2]{\overset{\lower0.1em\hbox{$\scriptscriptstyle #2$}}{#1}}

\usepackage[capitalise]{cleveref}

\usepackage{xspace}
\def\flow{RMP{flow}\xspace}
\def\algebra{RMP-algebra\xspace}
\def\tree{RMP-tree\xspace}

\def\pushforward{\texttt{pushforward}\xspace}
\def\pullback{\texttt{pullback}\xspace}
\def\resolve{\texttt{resolve}\xspace}

\newtheorem{theorem}{Theorem}[section]

\newtheorem{corollary}{Corollary}[theorem]
\newtheorem{lemma}[theorem]{Lemma}
\newtheorem{proposition}[theorem]{Proposition}

\begin{document}

% paper title
\title{Stable, Concurrent Controller Composition for Multi-Objective \\Robotic Tasks}

% You will get a Paper-ID when submitting a pdf file to the conference system
% \author{Author Names Omitted for Anonymous Review. Paper-ID 115}

\author{Anqi~Li$^\dagger$, Ching-An Cheng$^\dagger$, Byron Boots$^\dagger$, and Magnus~Egerstedt$^\dagger$
\thanks{*This work was sponsored in part by Grant No. W911NF-17-2-0181 from the U.S. Army Research Laboratory DCIST CRA.}% <-this % stops a space
\thanks{$^\dagger$Anqi Li, Ching-An Cheng, Byron Boots, and Magnus Egerstedt are with the Institute for Robotics and Intelligent Machines, Georgia Institute of Technology, Atlanta, GA 30332, USA. Email: {\tt\{anqi.li, cacheng, magnus\}@gatech.edu}, {\tt bboots@cc.gatech.edu}}% <-this % stops a space
}

\maketitle

\begin{abstract}

Robotic systems often need to consider multiple tasks concurrently. This challenge calls for controller synthesis algorithms that fulfill multiple control specifications %simultaneously
while maintaining the stability of the overall system. In this paper, we decompose multi-objective tasks into subtasks, where individual subtask controllers are designed independently  and then combined to generate the overall control policy. In particular, we adopt Riemannian Motion Policies (RMPs), a recently proposed controller structure in robotics, and, \flow, its associated computational framework for combining RMP controllers.
We re-establish and extend the stability results of \flow through a rigorous Control Lyapunov Function (CLF) treatment. We then show that \flow can
stably combine individually designed subtask controllers that satisfy certain CLF constraints. This new insight leads to an efficient CLF-based computational framework to generate stable controllers that consider all the subtasks simultaneously.
Compared with the original usage of \flow, our framework provides users the flexibility to incorporate design heuristics through nominal controllers for the subtasks.
We validate the proposed computational framework through numerical simulation and robotic implementation.

\end{abstract}

\IEEEpeerreviewmaketitle

\section{Introduction}

%In control of robotics systems, it is often required that multiple tasks are satisfied.

Multi-objective tasks are often involved in the control of robotic systems \cite{peters2008unifying,wang2016multi,ratliff2018riemannian,morris2013sufficient}.
For example, %a robot manipulator may need to achieve a certain end-effector configuration while avoiding collisions with obstacles, and
a group of robots may be tasked with achieving a certain formation, moving toward a goal region, while avoiding collisions with each other and  obstacles~\cite{wang2016multi}. These types of problems call for algorithms that can systematically generate a stable controller capable of fulfilling multiple control specifications simultaneously.

A classic strategy is to first design a controller for each individual control specification, and then provide a high-level rule to \emph{switch} among them.
This idea has been frequently exploited in robotics~\cite{arkin1998behavior}. For example, it is common practice to switch to a collision avoidance controller when the robot risks colliding with obstacles~\cite{arkin1998behavior}.
The stability of switching systems has been thoroughly investigated, e.g. by finding a common or switched Lyapunov function for the systems among all designed controllers~\cite{liberzon1999stability,vu2005common,narendra1994common,daafouz2002stability}.
However, a fundamental limitation shared by these switching approaches is that only a single controller is active at a time and hence only a subset of the control specifications is considered. If not designed properly, some controllers for secondary tasks might take over the operation for most of the time. For example, when a robot navigates in a cluttered environment, the collision avoidance controller can dominate other controllers and the primary tasks may never be considered~\cite{wang2016multi}. Therefore, it may be more desirable to
\emph{blend} controllers rather than impose a hard \emph{switch} between them, so that all tasks can be considered simultaneously.

In robotics, the strategy of weighting controllers for different tasks has been explored in potential field methods~\cite{KhatibPotentialFields1985,arkin1998behavior}. While easy to implement such schemes, it can be difficult to provide formal stability guarantees for the overall ``blended'' system, especially when the weights are state-dependent.
In some cases, the stability of the overall system has been shown through a common Lyapunov function~\cite{vu2005common,narendra1994common}, but the existence of a common Lyapunov function is not guaranteed. % depends on the properties of the individually designed controllers.
Finding a common Lyapunov function can be particularly challenging for robotics applications because the tasks can potentially conflict, e.g. the robot may need to move through a cluttered environment to go to the goal.

The framework of null-space or hierarchical control handles this problem by assigning priorities to the tasks, and hence to the controllers~\cite{peters2008unifying,escande2014hierarchical}. The performance of the high-priority tasks can be guaranteed by forcing the lower-priority controllers to act on the null space of high-priority tasks.
% This provides performance guarantees for the high-priority task, but s
However, several problems surface as the number of tasks increases.
One problem is the algorithmic singularities introduced by the usage of multiple levels of projections \cite{peters2008unifying,escande2014hierarchical}. Most algorithms are designed under the assumption of singular-free conditions. But this assumption is unlikely to hold in practice, especially when there are a large number of tasks, and the system can easily become unstable if the algorithmic singularities occur.
In addition, similar to the switching scheme, it is possible that secondary controllers, e.g. collision avoidance controllers, become the ones with high-priorities and the primary task can not be achieved.
While several heuristics~\cite{dietrich2012continuous,lee2012intermediate} have been proposed to shift the control priorities dynamically,
whether such systems can be globally stabilized in presence of the algorithmic singularities is still an open question~\cite{dietrich2018hierarchical}.

Control Lyapunov functions (CLFs) and control barrier functions (CBFs) constitute another class of methods to encode multiple control specifications~\cite{morris2013sufficient,ames2014control,wang2016multi}. In the CLF and CBF frameworks, the control specifications are encoded as constraints on the time derivatives of Lyapunov or barrier function candidates, and a control input that satisfies all the constraints is solved through a constrained optimization problem. However, in the case of conflicting specifications, the CLF and CBF frameworks %may not generate a feasible solution
suffer from feasibility problems~\cite{squires2018constructive}, i.e. there does not exist any controller that satisfies all the control specifications. Although the CLF constraints %on the time derivatives of Lyapunov function candidates
can be relaxed through slack variables~\cite{ames2014control}, they also add a new set of hyperparameters to trade off the importance of different specifications; care must be taken in tuning these hyperparameters in order to achieve desired performance properties and maintain stability.
Finally, it can be hard to encode certain high-dimensional control specifications, such as damping behaviors,
as CLF or CBF constraints.

In this paper, we focus on
\emph{weighting} individual controllers. We aim to address two interrelated questions:% while exploiting the advantages of the blending schemes, w
\begin{itemize}
    \item How can controllers be composed while guaranteeing system stability?
    \item How should individual controllers be designed so that they can be easily combined?
\end{itemize}

Although ensuring stability is challenging for arbitrary blending schemes, we design a systematic process to combine controllers so that the stability of the overall system is guaranteed.
Our framework considers all control specifications simultaneously, while providing the flexibility to vary the importance of different controllers based on the robot state.
Moreover, instead of considering specifications in the configuration space,
we allow for controllers defined directly on different spaces or manifolds\footnote{Specifications defined on non-Euclidean manifolds are common in robotics; for example, in obstacle avoidance, obstacles become holes in the space and the geodesics flow around them~\cite{ratliff2018riemannian}. } for different specifications. %  For example, the obstacle avoidance controller can be defined on a $1$-dimensional space that describes the distance between the robot and the obstacle, while

This separation can largely simplify the design and computation of each individual controller, because it only concerns a possibly lower-dimensional manifold that is directly relevant to a particular control specification. For example, controllers for different links of a robot manipulator can be designed in their corresponding (possibly non-Euclidan) workspaces.
% Specifications defined on non-Euclidean manifolds are common in robotics; for example, in obstacle avoidance, obstacles become holes in the space and the geodesics flow around them~\cite{ratliff2018riemannian}.
%
We leverage a recent approach to controller synthesis in robotics,
Riemannian Motion Policies (RMPs)~\cite{ratliff2018riemannian} and \flow~\cite{cheng2018rmpflow}, which have been successfully deployed on robot manipulators~\cite{cheng2018rmpflow,ratliff2018riemannian} and multi-robot systems~\cite{li2019multi}.
An RMP is a mathematical object that is designed to describe a controller on a manifold, and \flow is a computational framework for combining RMPs designed on different task manifolds into a controller for the entire system.
A particular feature of \flow is the use of state-dependent importance weightings of controllers based on the properties of the corresponding manifolds.
% An RMP pairs a controller with a positive semi-definite  matrix, which can be a function the state of a robot, to characterize the importance of the controller based on the properties of the manifolds.
It is show in ~\cite{cheng2018rmpflow} that when RMPs are generated from  Geometric Dynamical Systems (GDSs), the combined controller is Lyapunov-stable.

RMPs and \flow were initially studied in terms of the geometric structure of second-order differential equations~\cite{cheng2018rmpflow}, where Riemannian metrics on manifolds (of GDSs) naturally provide a geometrically-consistent notion of task importance and hence a mechanism to combine controllers (i.e. \flow).
While differential geometry provides a mathematical interpretation of \flow, in practice, the restriction to GDSs for control specifications could limit performance and make controller design  difficult.

To overcome this limitation, we revisit \flow with a rigorous CLF treatment and show that the existing computational framework of \flow actually ensures stability for a \emph{larger} class of systems than GDSs.
This discovery is made possible by an alternative stability analysis of \flow and an induction lemma that characterizes how the  stability of individual controllers is propagated to the combined controller in terms of CLF constraints.
Hence, we can reuse \flow to stably combine a range of controllers, not limited to the ones consistent with GDSs. %, with performance guarantees. %, to satisfy control specifications while ensuring stability.
% this could also provide an explanation to the previous heuristic use of \flow.
%
% To demonstrate, we use \flow to stably combine a class of subtask controllers satisfying the CLF constraints~\cite{ames2014control,morris2013sufficient}.
To demonstrate, we introduce a computational framework called \emph{\flow--CLF}, where we augment \flow with CLF constraints to generate a stable controller given user-specified nominal controllers for each of the control specifications. This allows  users to incorporate additional design knowledge given by, e.g. heuristics, motion planners, and human demonstrations, without worrying about the geometric properties of the associated manifolds. \flow--CLF can be viewed as a soft version of the QP--CLF framework~\cite{morris2013sufficient} that guarantees the stability of the overall system, while ensuring feasibility even when control specifications are  conflicting.

% The rest of the papers is organized as follows. In \cref{sec:background}, we briefly review RMPs and CLFs as two ways to achieve multiple control specifications simultaneously. Then \cref{sec:generalized_stability} provides a new stability analysis of \flow based on a formal CLF treatment. In \cref{sec:computational_framework} we introduce a computational framework for controller synthesis using \flow with CLF constraints. Finally, \cref{sec:results} and \cref{sec:conclusions} present the experimental results and the conclusions, respectively.

\section{Background}\label{sec:background}

We first review \flow~\cite{ratliff2018riemannian,cheng2018rmpflow} and CLFs~\cite{ames2014control,morris2013sufficient}, which are different ways to combine control specifications. %and

\subsection{Riemannian Motion Policies (RMPs) and \flow}\label{sec:rmp}

Consider a robot with configuration space $\CC$ which is a smooth $d$-dimensional manifold.
% For compactness of writing, without loss of generality, w
We assume that $\CC$ admits a global %\footnote{In local coordinates, the following derivations apply under proper change of coordinates. \todo{still not clear enough, what does global mean}}
\emph{generalized coordinate} $\q: \CC \to \R^d$ and follow the assumption in~\cite{cheng2018rmpflow} that the system can be feedback linearized in such a way that it can be controlled directly through the generalized acceleration\footnote{This setup can be extended to torque controls as in~\cite{peters2008unifying}.}, i.e. $ \qdd = \ub(\q,\qd)$.
We call $\ub$
% $\ub: \R^d \times \R^d \to \R^d$
a \emph{control policy} or a \emph{controller}, and $(\q, \qd)$ the \emph{state}.

The task is often defined on a different manifold from $\CC$ called the \emph{task space}, denoted $\TT$. %, which we assume is related to the configuration space $\CC$ through a smooth \emph{task map} $\psi: \CC \to \TT$.
A task may admit further structure as a composition of \emph{subtasks} (e.g. reaching a goal, avoiding collision with obstacles, etc.). In this case, we can treat the task space as a collection of multiple lower-dimensional \emph{subtask spaces}, each of which is a manifold. % of a subtask.
In other words, each subtask space is associated with a control specification and together the task space $\TT$ describes the overall multi-objective control problem.

\begin{figure}
    \centering
    \vspace{3mm}
    \resizebox{!}{2in}{\includegraphics{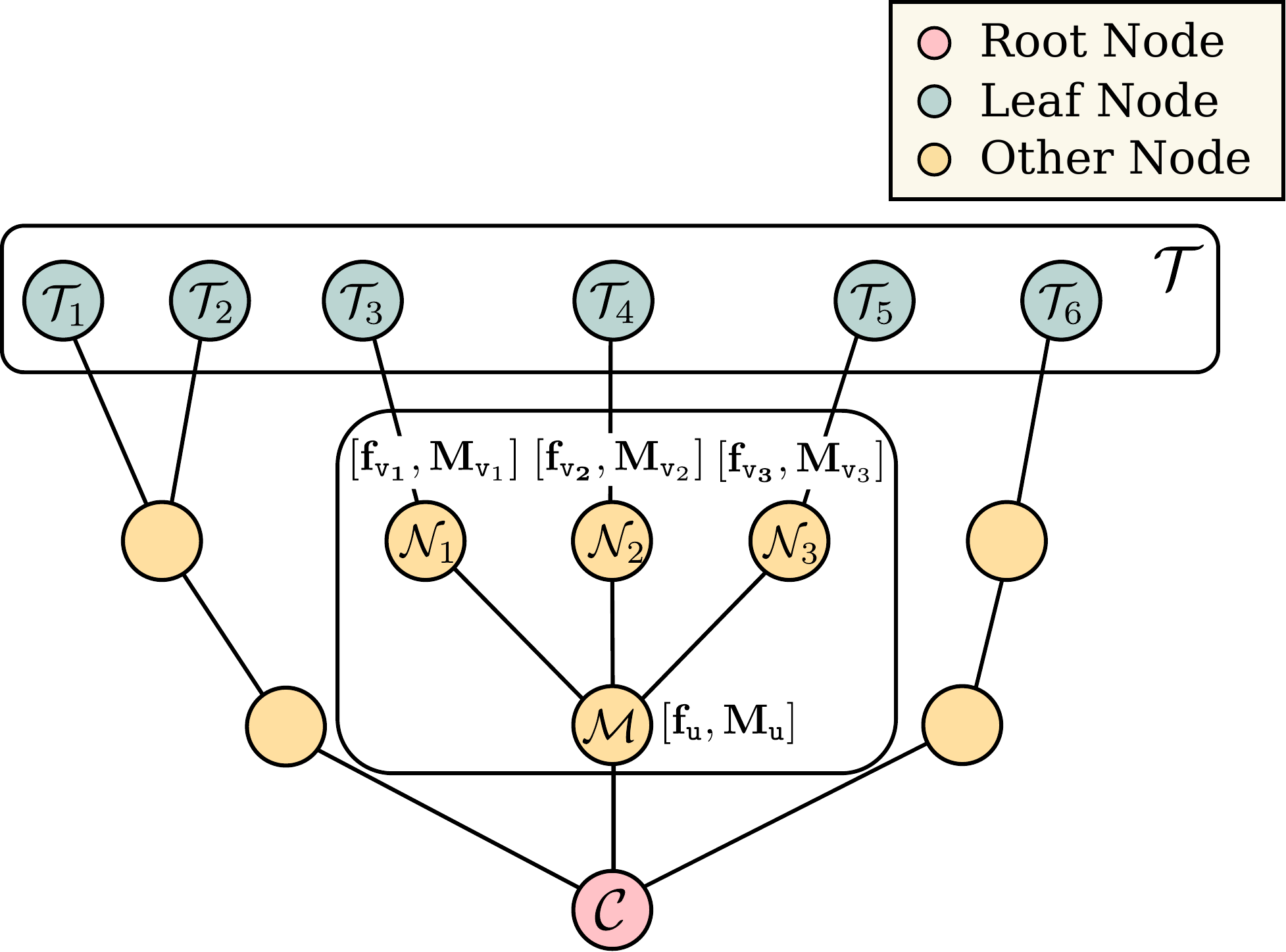}}
    \caption{An example of an \tree. See text for details. \vspace{-3mm}}
    \label{fig:rmp_tree}
\end{figure}

Ratliff et al. \cite{ratliff2018riemannian} propose \emph{Riemannian Motion Policies} (RMPs) to represent control policies on manifolds. Consider an $m$-dimensional manifold $\MM$ with a global coordinate $\x \in \R^m$. % Assume that the acceleration $\xdd$ can be directly controlled, i.e. $\ub=\xdd$.
An RMP on $\MM$ can be represented by two forms, its \emph{canonical form} $(\ab, \M)^\MM$ and its \emph{natural form}  $[\f, \M]^\MM$, where $\ab: (\x,\xd)\mapsto \ab(\x,\xd)$ is the desired acceleration, $\M: (\x,\xd)\mapsto \M(\x,\xd)\in\R_+^{m\times m}$ is the inertial matrix, and
$\f = \M \ab$ is the desired force. It is important to note that $\M$ and $\f$ do not necessarily correspond to physical quantities; $\M$ defines the importance of an RMP when combined with other RMPs, and $\f$ is proposed for computational efficiency.%  and stability of \flow. \al{\flow is not defined here...}

\flow~\cite{cheng2018rmpflow} is a recursive algorithm to generate control policies on the configuration space given the RMPs of subtasks.
It introduces: 1) a data structure, the \emph{\tree}, for computational efficiency;
and 2) a set of operators, the \emph{\algebra}, to propagate information across the \tree.

An \tree is a directed tree, which encodes the computational structure of the task map from $\CC$ to $\TT$ (see~\cref{fig:rmp_tree}).
In the \tree, a node is associated with the state and the RMP on a manifold, and an edge is augmented with a smooth map from a parent-node manifold to a child-node manifold.
In particular, the root node $\rtt$ is associated with the state of the robot $(\q,\qd)$ and its control policy on the configuration space $(\ab_\rtt,\M_\rtt)^\CC$, and each leaf node $\ltt_k$ is associated with the RMP $(\ab_{\ltt_k},\M_{\ltt_k})^{\TT_{k}}$ for a subtask, where $\TT_{k}$ is a subtask manifold. Recall the collection $\{\TT_{k}\}_{k=1}^K$ is the task space $\TT$, where $K$ is the number of tasks.

To illustrate how the \algebra operates, consider a node $\utt$ with $N$ child nodes $\{\vtt_j \}_{j=1}^N$. Let $\ett_j $ denote the edge from $\utt$ to $\vtt_j $ and let $\psi_{\ett_j}$ be the associated smooth map. Suppose that $\utt$ is associated with an RMP $[\f_\utt, \M_\utt ]^\MM$ on a manifold $\MM$ with coordinate $\xb$, and $\vtt_j$ is associated with an RMP $[\f_{\vtt_j}, \M_{\vtt_j} ]^{\NN_j}$ on a manifold $\NN_j$ with coordinate $\y_j$. (Note that here we represent the RMPs in their natural form.) The \algebra consists of the following three operators:

\begin{enumerate}
\item  \pushforward is the operator to forward propagate the \emph{state} from the parent node $\utt$ to its child nodes $\{\vtt_j \}_{j=1}^N$. Given the state $(\x,\xd)$ from $\utt$, it computes $(\y_j, \yd_j) = (\psi_{\ett_j}(\x) , \J_{\ett_j} (\x)\,\xd )$ for each child node $v_j$, where $\J_{\ett_j} = \partial_\x \psi_{\ett_j}$ is the Jacobian matrix of $\psi_{\ett_j}$.

\item \pullback is the operator to backward propagate the RMPs from the child nodes to the parent node. Given $\{[\f_{\vtt_j}, \M_{\vtt_j} ]^{\NN_j}\}_{j=1}^N$ from the child nodes, the RMP $[\f_{\utt}, \M_{\utt} ]^{\MM}$ for the parent node $\utt$ is computed as,% \vspace{-2mm}
\begin{equation*}\small%\label{eq:natural-pullback}\small
\f_\utt = \sum_{j=1}^N \J_{\ett_j}^\t (\f_{\vtt_j} - \M_{\vtt_j} \Jd_{\ett_j} \xd),\,\, \M_\utt = \sum_{j=1}^N \J_{\ett_j}^\t \M_{\vtt_j} \J_{\ett_j}.
\end{equation*}

% It can be shown that the canonical form $(\ub, \M)^{\MM}$ of the above RMP is the solution to the least-squares problem,
% \begin{align} \label{eq:least-square-pullback}
% \ub
% &=
% 	\argmin_{\ub'} \frac{1}{2} \sum_{j=1}^{K} \norm{\J_j \ub' + \dot\J_j \xd -  \ub_{j} }_{\M_j}^2   %\\
% \end{align}
% where $ \ub_j = \M_j^{\dagger}\,\f_j$, and $\norm{\cdot}_{\M_j}^2 = \lr{\cdot}{\M_j \cdot}$, $\dagger$ denotes the Moore-Penrose inverse.

\item \resolve maps an RMP from its natural form to its canonical form. Given $[\f_\utt, \M_\utt]^{\MM}$, it outputs $(\ab_\utt, \M_\utt)^{\MM}$ with $\ab_\utt = \M^{\dagger}\,\f_\utt$, where $\dagger$ denotes Moore-Penrose inverse.
\end{enumerate}

\flow performs control policy generation through
running the \algebra on the \tree.
It first performs a forward pass, by recursively calling \pushforward from the root node to the leaf nodes to update the state associated with each node on the \tree. Second, every leaf node $\ltt_k$ \emph{evaluates} its natural form RMP $\{(\fb_{\ltt_k},\M_{\ltt_k})^{\TT_{\ltt_k}}\}_{k=1}^K$ given its associated state. Then, \flow performs a backward pass, by recursively calling \pullback from the leaf nodes to the root node to back propagate the RMPs in the natural form. Finally, \resolve is applied to the root node to transform the RMP $[\f_\rtt, \M_\rtt]^\CC$ into its canonical form $(\ab_\rtt, \M_\rtt)^\CC$ and set the control policy as $\ub = \ab_\rtt$.

\flow was originally analyzed based on a differential geometric interpretation. Cheng et al.~\cite{cheng2018rmpflow} consider the inertial matrix $\M$ generated by a Riemannian metric on the \emph{tangent bundle} of the manifold $\MM$ (denoted as $T\MM$). Let $\Gb: (\x,\xd)\mapsto \Gb(\x,\xd)\in\R^{m\times m}_{+}$ be a (projected) Riemannian metric and define %It introduces
the \emph{curvature terms}
\vspace{-2mm}
\begin{equation}\label{eq:curvatures}%\small
    \begin{split}
    \bm\Xi_{\G}(\x,\xd)&\,\coloneqq\,\frac{1}{2} \sum_{i=1}^m \, \dot{x}_i\,\partial_{\xd}\, \gb_{i}(\x,\xd),\\
    \bm\xi_{\G}(\x,\xd)&\,\coloneqq\,\sdot{\Gb}{\x}(\x,\xd)\,\xd - \frac{1}{2} \nabla_\x\, (\xd^\t \Gb(\x,\xd)\, \xd),
    \end{split}
\end{equation}
where $\sdot{\Gb}{\xb}(\x,\xd) \coloneqq  [\partial_{\x} \, \gb_{i} (\x,\xd)\, \xd]_{i=1}^m$, $\gb_{i}(\x,\xd)$ is the $i$th column of $\Gb(\x,\xd)$, and $x_i$ is the $i$th component of $\x$. The inertial matrix $\M(\x,\xd)$ is then related to $\Gb(\x,\xd)$ through,% \vspace{-1mm}
\begin{equation}\label{eq:M}% \small
    \M(\x,\xd)=\Gb(\x,\xd) + \bm\Xi_{\G}(\x,\xd).
\end{equation}

Under this geometric interpretation, RMPs on a manifold $\MM$ can be (but not necessarily) generated from a class of systems called \emph{Geometric Dynamical Systems}~(GDSs)~\cite{cheng2018rmpflow}, whose dynamics are on the form of
\begin{equation}\label{eq:GDS}% \small
 \Mb(\x,\xd)\,\xdd
+ \bm\xi_{\G}(\x,\xd)  = - \nabla_\x \Phi(\x) - \Bb(\x,\xd)\,\xd,
\end{equation}
where $\B: \R^m \times \R^m \to \R^{m\times m}_{+}$ is the \emph{damping matrix}, and $\Phi: \R^m \to \R$ is the \emph{potential function}. %A GDS can be represented by a tuple $(\MM, \Gb, \B, \Phi)$.
When $\G(\x,\xd) = \G(\x)$, the GDSs reduce to the widely studied \emph{Simple Mechanical Systems}~\cite{bullo2004geometric}. % including spring-mass-damper systems. % An RMP $(\ab, \M)^\MM$ is naturally associated with the GDS by solving for $\ab=\xdd$ from (\ref{eq:GDS}) and computing $\Mb$ through (\ref{eq:M}).

The stability properties of \flow is analyzed in~\cite{cheng2018rmpflow} under the assumption that \emph{every} leaf-node RMP is specified as a GDS~\eqref{eq:GDS}. Before stating the stability theorem, let us define the metric, damping matrix, and potential function for every node in the \tree: For a leaf node, its metric, damping matrix, and potential are defined naturally by its underlying GDS. For a non-leaf node $\utt$ with $N$ children $\{\vtt_j\}_{j=1}^N$, these terms are  defined recursively by the relationship,
\begin{equation}\label{eq:mdp}\footnotesize
    \begin{split}
        \Gb_\utt = \sum_{j=1}^N\J_{\ett_j}^\t\G_{\vtt_j}\J_{\ett_j},\,
        \B_\utt = \sum_{j=1}^N\J_{\ett_j}^\t \B_{\vtt_j} \J_{\ett_j},\,
        \Phi_\utt =  \sum_{j=1}^N\Phi_{\vtt_j} \circ \psi_{\ett_j},
    \end{split}
    \vspace{-2em}
\end{equation}
where $\G_{\vtt_j}$, $\B_{\vtt_j}$ and $\Phi_{\vtt_j}$ are the metric, damping matrix, and potential function for the $j$th child. %$\psi_{\ett_j}$ and $\J_{\ett_j}$ are the smooth map associated with $j$th edge and its Jacobian matrix.% Note, however, that a non-leaf node generally does \emph{not} follow the GDS given by its corresponding metric, damping matrix, and potential. We refer the readers to \cite{cheng2018rmpflow} for more formal treatments of these concepts.
The stability results for \flow are stated below.
\begin{theorem}\label{thm:stability}
\textnormal{\cite{cheng2018rmpflow}}
Let $\G_\rtt$, $\B_\rtt$, and $\Phi_\rtt$ be the metric, damping matrix, and potential function of the root node defined in (\ref{eq:mdp}). If each leaf node is given by a GDS, $\G_\rtt, \B_\rtt \succ 0 $, and $\Mb_\rtt$ is non-singular, then the system converges to a forward invariant set $\CC_\infty \coloneqq \{(\q,\qd) : \nabla_\q \Phi_\rtt = 0, \qd = 0 \}$.
\end{theorem}

\subsection{Control Lyapunov Functions (CLFs)}\label{sec:clf}

Control Lyapunov Function (CLF) methods~\cite{ames2014control,morris2013sufficient,sontag1983lyapunov}  encode control specifications as Lyapunov function candidates. In these methods, controllers are designed to satisfy the inequality constraints on the time derivative of the Lyapunov function candidates.

 % Therefore, it poses inequality constraints on the control input. In the case when the system is control-affine and the Lyapunov function candidate has relative degree $1$, the time-derivative constraints becomes linear constraints on the control input. When augmented with a quadratic objective function on the control input (e.g. minimizing control effort, minimizing distance to some nominal controllers), the problem becomes a quadratic programming (QP) problem and can be solved efficiently.
Consider a dynamical system in control-affine form,
\begin{equation}\label{eq:control_affine}
    \dot{\bm\eta}\,=\,f(\bm\eta) + g(\bm\eta)\,\ub,
\end{equation}
where $\bm\eta\in\R^n$ and $\ub\in\R^m$ are the state and control input for the system.
We assume that $f$ and $g$ are locally Lipschitz continuous, and the system~\eqref{eq:control_affine} is forward complete, i.e. $\bm\eta(t)$ is defined for all $t\geq 0$.
For second-order systems considered by \flow, we have $\bm\eta=[\x^\top\,\,\xd^\top]^\top$,
\begin{equation}
    f(\bm\eta)\equiv\begin{bmatrix}
    \zerob & \Ib \\
    \zerob & \zerob
    \end{bmatrix},\quad
    g(\bm\eta)\equiv
    \begin{bmatrix}
    \zerob\\
    \Ib
    \end{bmatrix}.
\end{equation}

Suppose that a Lypapunov function candidate $V(\bm\eta)$ is designed for a control specification. The control input is then required to satisfy a CLF constraint, e.g. $\dot{V}\leq -\alpha(V)$, %A common CLF constraint is,
where $\alpha:\mathbb{R}_+\to\mathbb{R}_+$ is a locally Lipschitz class $\mathcal{K}$ function~\cite{khalil1996noninear} (i.e. $\alpha$ is strictly increasing and $\alpha(0)=0$).
In the case of control-affine system, the CLF constraint becomes a linear inequality constraint on control input $\ub$ given state $\bm\eta$,
\begin{equation}\label{eq:decay_linear}
    L_g V(\bm\eta) \, \ub \leq -L_f V(\bm\eta) -\alpha(V(\bm\eta)),
\end{equation}
where $L_f V$ and $L_g V$ are the \emph{Lie derivatives} of $V$ along $f$ and $g$, respectively.

 When there are multiple control specifications, one can design Lyapunov function candidates $\{V_k\}_{k=1}^{K}$ separately. Then the controller synthesis problem becomes finding a controller that satisfies all the linear inequalities given by the Lyapunov function candidates. Morris et al.~\cite{morris2013sufficient} propose a computational framework, QP--CLF, % to generate controllers that satisfy all control specifications simultaneously.
that solves for the controller through a Quadratic programming (QP) problem that augments the constraints with a quadratic objective:
\begin{equation}\label{eq:qp_clf}
\begin{aligned}
& \underset{\ub}{\min}& &
\frac{1}{2}\,\ub^\top
\,H(\bm\eta)\,\ub + F(\bm\eta)^\top\,\ub\\
& \text{s.t.}
& & L_g V_k(\bm\eta) \, \ub \leq -L_f V_k(\bm\eta) -\alpha_k(V_k(\bm\eta)),\\
& & & \qquad \forall\,k\in\{1,\ldots,K\}.
\end{aligned}
\end{equation}

However, when the specifications are conflicting, it may not be possible to enforce the CLF constraints for all $\{V_k\}_{k=1}^{K}$ since the optimization problem~\eqref{eq:qp_clf} can become infeasible~\cite{morris2013sufficient}. In~\cite{ames2014control}, Ames et al. introduce \emph{slack variables} $\{\delta_k\}_{k=1}^K$ so that the optimization problem is always feasible. Let $\bar{\ub} = [\ub^\top\,\,\delta_1\,\,\ldots\,\,\delta_{K}]^\top$ denote % the the ensemble of
all decision variables.
Then the relaxed optimization problem becomes,
\begin{equation}\label{eq:qp_clf_slack}
\begin{aligned}
& \underset{\bar{\ub}}{\min}& &
\frac{1}{2}\,\bar{\ub}^\top
\,\bar{H}(\bm\eta)\,\bar{\ub} + \bar{F}(\bm\eta)^\top\,\bar{\ub}\\
& \text{s.t.}
& & L_g V_k(\bm\eta) \, \ub \leq -L_f V_k(\bm\eta) -\alpha(V_k(\bm\eta)) + \delta_k,\\
& & & \qquad \forall\,k\in\{1,\ldots,K\},
\end{aligned}
\end{equation}
where $\bar{H}(\bm\eta)$ and $\bar{F}(\bm\eta)$ encode how the original objective function and the CLF constraints are balanced. However, %in this case, it is hard to guarantee stability for the overall system for general choices of $\bar{H}(\bm\eta)$ and $\bar{F}(\bm\eta)$.
care must be taken in tuning $\bar{H}(\bm\eta)$ and $\bar{F}(\bm\eta)$ to achieve desired performance properties and maintain stability.

\section{The CLF Interpretation of \flow}\label{sec:generalized_stability}

The goal of this paper is to combine control policies specified for subtask manifolds into a control policy for the robot with stability guarantees. \flow provides a favorable computational framework but its original analysis is limited to subtask control policies generated by GDSs.
This assumption is rather unsatisfying, as % it assumes that the GDSs for every leaf nodes are given by the \emph{users}. This means that
the users need to encode the control specifications as GDS behaviors. % This is nontrivial as the curvature terms in the GDSs~(\ref{eq:GDS}) are rather complicated.
Further, this restriction can potentially limit the performance of the subtasks and result in unnecessary energy consumption by the system.
Although empirically \flow has been shown to work with non-GDS leaf policies~\cite{cheng2018rmpflow}, it is unclear if the overall system is still stable.

In this section, we show that the \algebra actually preserves the stability of a wider range of leaf-node control policies than GDSs.
We relax the original GDS assumption in~\cite{cheng2018rmpflow} to a more general CLF constraint on each leaf node, and provide a novel stability analysis of \flow.
These results allow us to reuse \flow for combining a more general class of control policies, which we will demonstrate  by combining controllers based on CLF constraints.
%However, it is not obvious whether the original stability analysis presented in~\cite{cheng2018rmpflow} still applies, as it relies on a class of auxiliary dynamical systems derived from GDSs.

%

%In contrast to CLF methods, GDSs explicitly specify the control input (which in turn dictate the decay-rate of the Lyapunov function). This additional constraint of being GDSs, however, is not necessary for stability.

%called \emph{structured Geometric Dynamical Systems}~(structured GDSs) and establishes the stability results based on the stability properties of this class of systems. \todo{state what we need/want}

\subsection{An Induction Lemma}

In order to establish CLF constraints on leaf nodes that guarantee stability, we first need to understand how the \algebra, especially the \pullback operator, connects the stability results of the child nodes to the parent node.

Again, let us consider a node $\utt$ with $N$ child nodes $\{\vtt_j \}_{j=1}^N$, in which $\utt$ is associated with %an acceleration policy $\ab_\utt$ on
a manifold $\MM$ with coordinate $\xb$, and $\vtt$ is associated with % an acceleration policy $\ab_{\vtt_j}$ on
a manifold $\NN_j$ with coordinate $\y_j$. In addition, let %$\ett_j$ be the edge from $\utt$ to $\vtt_j$, which is augmented with
$\psi_{\ett_j}:\x\mapsto \y_j$ be the smooth map between manifolds $\MM$ and $\NN_j$. We furthur assume that $\psi_{\ett_j}$ is surjective, i.e., $\NN_j=\psi_{\ett_j}(\MM)$. %, and let $\J_{e_j}$ be the Jacobian matrix of $\psi_{e_j}$.

Let us associate each child node $\vtt_j$ with a proper, continuously differentiable and lower-bounded potential $\Phi_{\vtt_j}$ on its manifold $\NN_j$ along with a continuously differentiable Riemannian metric $\G_{\vtt_j}$ on its tangent bundle $\TT\NN_j$. Then, for node $\vtt_j$, there is a natural Lyapunov function candidate,\vspace{-2mm}
\begin{equation}\label{eq:lyap_cand_child}
V_{\vtt_j}(\y_j,\yd_j) = \frac{1}{2}\,\yd_j^\t\,\Gb_{\vtt_j}(\y_j,\yd_j)\,\yd_j + \Phi_{\vtt_j}(\y_j),
\end{equation}
and an associated natural-formed RMP $[\f_{\vtt_j}, \M_{\vtt_j} ]^{\NN_j}$, where $\f_{\vtt_j}$ is the force policy, and $\Mb_{\vtt_j}$ is defined by $\G_{\vtt_j}$ as in~(\ref{eq:M}). We shall further assume that $\M_{\vtt_j}$ is locally Lipschitz continuous for the ease of later analysis.
By construction of the \algebra, these Lyapunov function candidates and RMPs of the child nodes $\{\vtt_j\}_{j=1}^N$ define, for the parent node $\utt$,
%, a Riemannian metric on the tangent bundle $\TT\MM$ and a potential on the the manifold $\MM$
%\begin{equation}\label{eq:mdp
%    \Gb_\utt = \sum_{j=1}^N \J_{\ett_j}^\t\,\Gb_{\vtt_j}\,\J_{\ett_j}\mathrm{,\,and\,\,}\Phi_\utt = \sum_{j=1}^N \Phi_{\vtt_j}\circ \psi_{\ett_j}
%\end{equation}
%as well as
%Consequently, we define
a Lyapunov function candidate
\begin{equation}\label{eq:lyap_parent}
V_\utt(\x,\xd) = \frac{1}{2}\,\x^\t\,\Gb_\utt(\x,\xd)\,\xd + \Phi_\utt(\x).
\end{equation}
where $\Gb_\utt$ and $\Phi_\utt$ are given in~\eqref{eq:mdp}.
%which reflect the properties of all the children $\{\vtt_j\}_{j=1}^N$.

The following lemma states how the decay-rate of $V_\utt$ is connected to the decay-rates of $\{V_{\vtt_j}\}_{j=1}^N$ via \pullback.
\begin{lemma}\label{lm:induction_general}
For each child node $\vtt_j$, assume that $\fb_{\vtt_j} = \Mb_{\vtt_j}\,\ydd_j$ renders $\dot{V}_{\vtt_j}(\y_j,\yd_j) = -U_{\vtt_j}(\y_j,\yd_j)$ for $V_{\vtt_j}$ in~(\ref{eq:lyap_cand_child}).
If the parent node $\utt$ follows dynamics $\fb_\utt = \Mb_\utt(\x, \xd)\,\xdd$, where $\fb_\utt$ and $\Mb_\utt$ are given by \pullback,  % (\ref{eq:natural-pullback}),
then $\dot{V}_\utt(\x,\xd) = -\sum_{j=1}^N U_{\vtt_j}(\y_j,\yd_j)$ for %the Lyapunov function candidate
$V_\utt$ in~(\ref{eq:lyap_parent}).
\end{lemma}
\begin{proof}
For notational convenience, we suppress the arguments of functions.
First, note that,
\begin{equation}\small
    \begin{split}
        V_\utt &= \frac{1}{2}\,\sum_{j=1}^N \x^\t\, \J_{\ett_j}^\t\,\Gb_{\vtt_j}\,\J_{\ett_j}\,\xd + \sum_{j=1}^N\,\Phi_{\vtt_j} = \sum_{j=1}^N\, V_{\vtt_j}.
    \end{split}
\end{equation}
As $\Gb_{\vtt_j}$ is a function in both $\y$ and $\yd$, following a similar derivation as in~\cite{cheng2018rmpflow}, we can show
\begin{equation}\small
    \begin{split}
        &\dot{V}_\utt = \sum_{j=1}^N\,\yd_j^\t\,\Mb_{\vtt_j}\,\ydd_j + \frac{1}{2}\,\yd_j^\t\sdot{\Gb_{\vtt_j}}{\y_j}\,\yd_j + \yd_j^\t\,\nabla_{\y_j}\Phi_{\vtt_j},
    \end{split}
\end{equation}
where $\Mb_{\vtt_j}$ and $\sdot{\Gb_{\vtt_j}}{\y_j}$ are the inertial matrix and the curvature term defined in~\cref{sec:rmp}.

Note that $\ydd_j = \J_{\ett_j}\xdd + \Jd_{\ett_j}\xd$, where $\xdd$ is given by the RMP $[\fb_\utt, \Mb_\utt]^\MM$. Hence, the first term can be rewritten as
\begin{equation*}\small
    \begin{split}
        &\sum_{j=1}^N\,\yd_j^\t\,\Mb_{\vtt_j}\,\ydd_j = \xd^\t\,\left(\sum_{j=1}^N \J_{\ett_j}^\t\,\Mb_{\vtt_j}\,\left(\J_{\ett_j}\,\xdd + \Jd_{\ett_j}\,\xd\right)\right)\\[-0.2em]
        % &= \xd^\t\,\left(\sum_{j=1}^K \J_j^\t\,\Mb_j\,\J_j\right)\,\xdd + \xd^\t\,\left(\sum_{j=1}^K \J_j^\t\,\Mb_j\,\Jd_j\xd\right)\\
        &= \xd^\t\,\fb_\utt + \xd^\t\,\left(\sum_{j=1}^N \J_{\ett_j}^\t\,\Mb_{\vtt_j}\,\Jd_{\ett_j}\,\xd\right)\\[-0.2em]
        &= \xd^\t\,\left(\sum_{j=1}^N\,\J_{\ett_j}^\t\,(\f_{\vtt_j} - \M_{\vtt_j}\,\Jd_{\ett_j}\,\xd)\right)
        + \xd^\t\,\left(\sum_{j=1}^N \J_{\ett_j}^\t\,\Mb_{\vtt_j}\,\Jd_{\ett_j}\,\xd\right)\\[-0.2em]
        &= \xd^\t\,\left(\sum_{j=1}^N\,\J_{\ett_j}^\t\,\fb_{\vtt_j}\right)= \sum_{j=1}^N\,\yd_j^\t\fb_{\vtt_j}.
    \end{split}
\end{equation*}
The time-derivative of $V_{\utt}$ can then be simplified as
\begin{equation}\small
    \begin{split}
        \dot{V}_\utt &= \sum_{j=1}^N\,\yd_j^\t\,\fb_{\vtt_j} + \frac{1}{2}\,\yd_j^\t\sdot{\Gb_{\vtt_j}}{\y_j}\,\yd_j + \yd_j^\t\,\nabla_{\y_j}\Phi_{\vtt_j}.
    \end{split}
\end{equation}
By assumption on ${\vtt}_j$, we also have
\begin{equation}\small
    \yd_j^\t\,\fb_{\vtt_j} + \frac{1}{2}\,\yd_j^\t\sdot{\Gb_{\vtt_j}}{\y_j}\,\yd_j + \yd_j^\t\,\nabla_{\y_j}\Phi_{\vtt_j} = -U_{\vtt_j}(\y_j,\yd_j).
\end{equation}
The statement follows then from the two equalities.
%Therefore,
%\begin{equation}
%    \begin{split}
%        \dot{V}_\utt(\x,\xd) &= -\sum_{j=1}^N\,U_{\vtt_j}(\y_j, %\yd_j)
%    \end{split}
%\end{equation}
\end{proof}

We can use \cref{lm:induction_general} to infer the overall stability of \flow.
For an \tree with $K$ leaf nodes, let leaf node $\ltt_k$ be defined on task space $\TT_k$ with coordinates $\z_k$
%Suppose that we associate $\ltt_k$ with a proper, continuously differentiable, and lower-bounded potential $\Phi_{\ltt_k}$, and a continuously differentiable Riemannian metric $\G_{\ltt_k}$. Then it has a natural Lyapunov function candidate
and has a natural Lyapunov function candidate
\begin{equation}\label{eq:lyap_cand_leaf}
    V_{\ltt_k}(\z_k,\zd_k) = \frac{1}{2}\zd_k^\t\,\Gb_{\ltt_k}(\z_k,\zd_k)\,\zd_k + \Phi_{\ltt_k}(\z_k).
\end{equation}
for some potential function $\Phi_{\ltt_k}$ and positive-definite Riemannian metric $\G_{\ltt_k}$ defined as above.
By~\cref{lm:induction_general}, if each leaf node $\ltt_k$ satisfies a CLF constraint,
\begin{equation}\label{eq:lyap_decay_leaf_general}
    \dot{V}_{\ltt_k}(\z_k,\zd_k)=-U_{\ltt_k}\,(\z_k,\zd_k),
\end{equation}
then a similar constraint is satisfied by the root node. This observation is summarized below without proof.

\begin{proposition}\label{pr:lyap_decay_general}
For each leaf node $\ltt_k$, assume that $\fb_{\ltt_k} = \Mb_{\ltt_k}\,\zdd_k$ renders (\ref{eq:lyap_decay_leaf_general}) for %the Lyapunov function candidate
~(\ref{eq:lyap_cand_leaf}). Consider the Lyapunov function candidate at the root node $V_\rtt(\q,\qd)$ defined through~(\ref{eq:lyap_parent}). Then, for the root node control policy of \flow, it holds
$\dot{V}_{\rtt}(\q,\qd)=-\sum_{k=1}^{K}\,U_{\ltt_k}\,(\psi_{\rtt\to\ltt_k}(\q),\J_{\rtt\to\ltt_k}\qd),$
where $\psi_{\rtt\to\ltt_k}$ is the map from  $\CC$ to  $\TT_k$, which can be obtained through composing maps from the root node $\rtt$ to the leaf node $\ltt_k$ on the \tree, and $\J_{\rtt\to\ltt_k}=\partial_\q \psi_{\rtt\to\ltt_k}$.
\end{proposition}

\begin{comment}
\begin{proof}
By assumption, we have,
\begin{equation}
    \dot{V}_{\ltt_k}(\z_k,\zd_k)=-U_{\ltt_k}\,(\z_k,\zd_k)
\end{equation}
for all $k=1,\ldots,K$. Then, the proposition follows from applying~\cref{lm:induction_general} recursively from the leaf nodes to the root node.
\end{proof}
\end{comment}

Note that \cref{pr:lyap_decay_general} provides an alternative way to show the stability results of \flow.
\begin{corollary}\label{cor:lyap_decay_gds}
For each leaf node $\ltt_k$, assume that $\fb_{\ltt_k}$ is given by a GDS $(\TT_k,\G_{\ltt_k},\B_{\ltt_k},\Phi_{\ltt_k})$. Consider the Lyapunov function candidate at the root node $V_\rtt(\q,\qd)$ defined recursively through~(\ref{eq:lyap_parent}). Then we have,
$\dot{V}_\rtt(\q,\qd)=-\qd^\top\,\Bb_\rtt(\q,\qd)\,\qd$
under the resulting control policies from \flow, where $\Bb_\rtt$ is defined recursively through~(\ref{eq:mdp}).
\end{corollary}
\noindent With \cref{cor:lyap_decay_gds}, we then can show \cref{thm:stability} by invoking LaSalle's invariance principle~\cite{khalil1996noninear}.

\subsection{Global Stability Properties}

More importantly, by \cref{pr:lyap_decay_general}, we can find how CLF constraints are propagated from the leaf nodes to the root node through \pullback.

\begin{proposition}\label{cor:lyap_decay_clf}
For each leaf node $\ltt_k$, assume that $\fb_{\ltt_k}=\Mb_{\ltt_k}\zdd_{\ltt_k}$ renders, for ${V}_{\ltt_k}$ in \eqref{eq:lyap_cand_leaf},
\begin{equation}\label{eq:lyap_decay_clf_leaf}
    \dot{V}_{\ltt_k}(\z_k,\zd_k)\leq-\alpha_{k}\,(\|\zd_k\|),
\end{equation}
where $\alpha_{k}$ is a locally Lipschitz continuous class $\KK$ functions~\cite{khalil1996noninear}. Consider the Lyapunov function candidate at the root node $V_\rtt(\q,\qd)$ defined recursively through~(\ref{eq:lyap_parent}). Then
\begin{equation}\label{eq:lyap_decay_clf}
    \dot{V}_\rtt(\q,\qd)\leq -\sum_{k=1}^{K}\,\alpha_{k}\,(\|\J_{\rtt\to\ltt_k}\,\qd\|)
\end{equation}
under the resulting control policies from \flow.
\end{proposition}

With this insight, we  state a new stability theorem of \flow by applying LaSalle's invariance principle. We assume that the inertia matrix at the root node $\M_\rtt$ is nonsingular for simplicity, so that the actual control input can be solved through the \resolve operation; a sufficient condition for $\M_\rtt$ being nonsingular is provided in~\cite{cheng2018rmpflow}.

\begin{theorem}\label{thm:clf_rmpflow_stability}
%Assume that the configuration space $\CC$ is embedded in the task space $\TT$ and (\ref{eq:lyap_decay_clf_leaf}) holds for each $\fb_{\ltt_k}$.
For each leaf node $\ltt_k$, assume that $\fb_{\ltt_k}=\Mb_{\ltt_k}\zdd_{\ltt_k}$ renders~\eqref{eq:lyap_decay_clf_leaf}.
% Let $\M_\rtt$ be the inertia matrix at the root node defined recursively through~\eqref{eq:mdp}.
Suppose that $\M_\rtt$ is nonsingular, and
the task space $\TT$ is an immersion of the configuration space $\CC$.
Then the control policy generated by \flow renders the system converging to the forward invariant set
\begin{equation}\label{eq:forward_inv_set}
    \CC_\infty \coloneqq \left\{(\q,\qd) : \qd = 0,\,\, \sum_{j=1}^{K} \J_{\rtt\to\ltt_k}^\t\, \f_{\ltt_k}=0 \right\}.
\end{equation}

\noindent
Further if, for all leaf nodes, % $\{\ltt_k\}_{k=1}^{K}$,
$\f_{\ltt_k}=-\nabla_{\z_{k}}\Phi_{\ltt_k}(\z_{k})$ when $\zd_{k}=0$, the system converges to the forward invariant set
\begin{equation}\label{eq:forward_inv_set_phi}
    \CC_\infty^\Phi \coloneqq \{(\q,\qd) : \nabla_\q \Phi_\rtt(\q) = 0, \qd = 0 \},
\end{equation}
where $\Phi_\rtt$ is the potential in $V_\rtt$ defined recursively in~(\ref{eq:mdp}).
\end{theorem}

%$\G_\rtt\succ0$ so that the Lyapunov function candidate $V_\rtt$ is well-defined, and

\begin{proof}

By assumption, $V_\rtt$ is proper, continuously differentiable and lower bounded. Hence, the system converges to the largest invariant set in $\{(\q,\qd):\dot{V}_\rtt(\q,\qd)=0\}$ by LaSalle's invariance principle~\cite{khalil1996noninear}. By \eqref{eq:lyap_decay_clf} in \cref{cor:lyap_decay_clf}, $\dot{V}_\rtt=0$ if and only if $\J_{\rtt\to\ltt_k}\,\qd=0$ for all $k=1,\ldots,K$. Since $\CC$ is immersed in $\TT$, we have $\qd=0$. Hence, the system converges to a forward invariant set $\CC_\infty\coloneqq\{(\q,\qd) : \qd= 0 \}$. Any forward invariant set in $\CC_\infty$ must have $\qdd = 0$, which implies that $\f_\rtt = 0$ as $\M_\rtt$ is nonsingular. Note that $\f_\rtt$ is given by the \pullback operation % ~(\ref{eq:natural-pullback})
recursively, hence,
\begin{equation*}\small
    0=\f_\rtt = \sum_{k=1}^{K} \J_{\rtt\to\ltt_k}^\t\, (\f_{\ltt_k} - \M_{\ltt_k}\,\Jd_{\rtt\to\ltt_k}\,\qd) = \sum_{k=1}^{K} \J_{\rtt\to\ltt_k}^\t\, \f_{\ltt_k}
\end{equation*}
where the last equality follows from $\qd= 0$.
Thus, the system converges to the forward invariant set in~\eqref{eq:forward_inv_set}.

Now, assume that $\f_{\ltt_k}=-\nabla_{\z_k}\Phi_{\ltt_k}(\z_{k})$ when $\zd_{k}=0$ (which is implied by $\qd=0$). Notice that by the definition of $\Phi_\rtt$ in (\ref{eq:mdp}), $\Phi_{\rtt}(\q) = \sum_{k=1}^{K}\Phi_{\ltt_k}(\z_k)$. By the chain rule,
\begin{equation*}\small
        \nabla_\q\Phi_\rtt(\q) = \sum_{k=1}^{K} \J_{\rtt\to\ltt_k}^\t\, \nabla_{\z_{k}}\Phi_{\ltt_k}(\z_{k})= -\sum_{k=1}^{K} \J_{\rtt\to\ltt_k}^\t\, \f_{\ltt_k}.
\end{equation*}
Hence $\sum_{k=1}^{K} \J_{\rtt\to\ltt_k}^\t\, \f_{\ltt_k}=0$ implies $\nabla_\q\Phi_\rtt(\q)=0$. Thus, the system converges forwardly to \eqref{eq:forward_inv_set_phi}. %the forward invariant set $\CC_\infty^\Phi = \{(\q,\qd):  \nabla_\q \Phi_\rtt(\q) = 0, \qd = 0 \}$.
\end{proof}

\cref{thm:clf_rmpflow_stability} implies that subtask controllers satisfying CLF constraints~\eqref{eq:lyap_decay_clf_leaf} can be stably combined by \flow. %Based on this result, we propose a computational framework in the following section to design and combine subtask controllers with stability guarantees.

% A remark is that although the analysis above is stated from time-invariant Lyapunov candidates, it can be generalized to time-varying Lyapunov candidates with time-varying metrics and potentials. To see this,
% \begin{equation*}
%     \dot{V}(\x,\xd, t) = \frac{\partial V}{\partial \xd}\xdd + \frac{\partial V}{\partial \x}\xd + \frac{\partial V}{\partial t}
% \end{equation*}

% \todo{add a note about time-varying lyapunov functions}

\section{A Computational Framework for \flow with CLF Constraints}\label{sec:computational_framework}

We introduce a computational framework for controller synthesis based on the stability results presented in~\cref{sec:generalized_stability}.
The main idea is to leverage \cref{cor:lyap_decay_clf}, which says that \flow is capable of preserving CLF constraints in certain form.
Recall that for leaf node $\vtt_j$, the constraint on the time-derivative of the Lyapunov function in \cref{cor:lyap_decay_clf} is
$\dot{V}_{\ltt_k}(\z_k,\zd_k)\leq-\alpha_{k}\,(\|\zd_k\|)$. Combined with the particular choice of leaf-node Lyapunov function candidate in~\eqref{eq:lyap_cand_leaf}, this yields a CLF constraint\footnote{This a linear constraint with respect to $\f_{\ltt_k}$. When $\zd_k=0$, % the Lyapunov function candidate $V_{\ltt_k}$ has a relative degree greater than $1$,
the constraint \eqref{eq:leaf_node_constraints} holds trivially because both sides equal $0$.}
\begin{equation}\label{eq:leaf_node_constraints}
    \begin{split}
        \zd_k^\top\,\f_{\ltt_k}\,&\leq\,-\zd_k^\top\left(\nabla_{\z_k}\Phi_{\ltt_k}+\bm\xi_{\Gb_{\ltt_k}}\right) - \alpha_{k}(\|\zd_k\|),
    \end{split}
\end{equation}
where $\bm\xi_{\G_{\ltt_k}}$ is defined in \eqref{eq:curvatures}.
\cref{cor:lyap_decay_clf} shows that, when the leaf-node control policies satisfy \eqref{eq:leaf_node_constraints}, \flow will yield a stable controller under suitable conditions. This provides a constructive principle to synthesize controllers.

\subsection{Algorithm Details}

Assume that some nominal controller $\ub_{\ltt_k}^d$ is provided by the specification of subtask $k$. We design the leaf-node controller as a \emph{minimally invasive} controller that modifies the nominal controller as little
as possible while satisfying the CLF  constraint~\eqref{eq:leaf_node_constraints}:
\begin{align}\label{eq:qp}
\f_{\ltt_k}^* & = & & {\textstyle \argmin_{\f_{\ltt_k}} } \,\|\f_{\ltt_k} - \Mb_{\ltt_k}\,\ub_{\ltt_k}^d \|_{\Pb_{\ltt_k}}^2    \\
 & \mathrm{s.t.}
 & & \zd_k^\top\,\f_{\ltt_k}\,\leq\,-\zd_k^\top\left(\nabla_{\z_k}\Phi_{\ltt_k}+\bm\xi_{\Gb_{\ltt_k}}\right) - \alpha_{k}(\|\zd_k\|) \nonumber
\end{align}
where $\Pb_{\ltt_k}\succ 0$ and $\M_{\ltt_k}$ is given by $\G_{\ltt_k}$ through~\eqref{eq:M}.
%is a positive definite matrix that measures the difference between $\f_{\ltt_k}$ and the natural-formed nominal controller $\fb_{\ltt_k}^d = \Mb_{\ltt_k}\,\ub_{\ltt_k}^d$.
Possible choices of $\Pb_{\ltt_k}$ include the identity matrix $\I$ and the inverse of the inertial matrix $\Mb_{\ltt_k}^{-1}$. In particular,  $\Pb_{\ltt_k}=\Mb_{\ltt_k}^{-1}$ yields an objective function equivalent to $\|\ab_{\ltt_k} - \ub_{\ltt_k}^d \|_{\Mb_{\ltt_k}}^2$, where $\ab_{\ltt_k}$ is the acceleration policy of the node. % Sometimes it is convenient to design the nominal controllers in the natural form $\fb_{\ltt_k}^d$ rather than the canonical form $\ub_{\ltt_k}^d$. \boots{This last sentence seems like a disconnected thought. Either delete or expand?}

We combine this minimally invasive controller design with \flow as a new computational framework for controller synthesis, called \emph{\flow--CLF}. \flow--CLF follows the same procedure as the original \flow~\cite{cheng2018rmpflow} as is discussed in \cref{sec:rmp}. The difference is that the leaf nodes solve for the RMPs based on the optimization problem~\eqref{eq:qp} during the evaluation step.
Note that \eqref{eq:qp} is a QP problem with a single linear constraint, so it can be solved analytically by projecting $\Mb_{\ltt_k}\ub_{\ltt_k}^d$ \mbox{onto the half-plane given by the  constraint.}

% \flow--CLF first performs a forward pass: it recursively calls \pushforward from the root node to the leaf nodes to update the state information in each node in the \tree. Secondly, every leaf node $\ltt_k$ solves~\eqref{eq:qp} to get its natural form RMP $\{(\fb_{\ltt_k},\M_{\ltt_k})^{\TT_{k}}\}$. Then, \flow--CLF performs a backward pass: it recursively calls \pullback from the leaf nodes to the root node to back propagate the RMPs in the natural form. Finally, \flow--CLF calls \resolve at the root node to transform the RMP $[\f_\rtt, \M_\rtt]^\CC$ into its canonical form $(\ab_\rtt, \M_\rtt)^\CC$ and solve for the control input $\ub =\ab_\rtt$.

\subsection{Stability Properties}

The form of~\eqref{eq:qp} together with \cref{thm:clf_rmpflow_stability} and the results of~\cite{morris2013sufficient} yields the following theorem:
\begin{theorem}\label{thm:clf_rmpflow_stability_comp}
%Assume that the map from the configuration space $\CC$ is embedded in the task space $\TT$ and each $\fb_{\ltt_k}$ is given by~\eqref{eq:qp} where each $\ub_{\ltt_k}^d$ is locally Lipschitz continuous. If $\M_\rtt$ is nonsingular and all edge Jacobians $\{J_{\ett_j}\}_j$ as well as all leaf inertial matrices $\{\Mb_{\ltt_k}\}_k$ are locally Lipschitz continuous, the control policy generated by \flow is locally Lipschitz continuous and renders the system converging to the forward invariant set $\CC_\infty$ defined in~\eqref{eq:forward_inv_set}.
Under the assumptions in \cref{thm:clf_rmpflow_stability}, if $\{\ub_{\ltt_k}^d\}_{k=1}^K$, $\{\Mb_{\ltt_k}\}_{k=1}^K$, all edge Jacobians and their derivatives are locally Lipschitz continuous, then the control policy generated by \flow--CLF is locally Lipschitz continuous and renders the system converging forwardly to~\eqref{eq:forward_inv_set}.
\end{theorem}

\begin{proof}
By \cref{thm:clf_rmpflow_stability}, the system converges to~\eqref{eq:forward_inv_set}. By~\cite{morris2013sufficient}, for all $k\in\{1,\ldots, K\}$, $\f_{\ltt_k}$ is locally Lipschitz. Since under the assumption \pullback and \resolve preserves Lipschitz continuity; the statement follows.
%Since $\M_\rtt$ is nonsingular, and all the Jacobian matrices as well as all leaf-node inertial matrices are Lipschitz continuous, the \pullback and \resolve operations preserve locally Lipschitz continuity. Therefore, the resulting controller is locally Lipschitz continuous.
\end{proof}

% If we want to guarantee that the system converges to the forward invariant set $\CC_\infty^\Phi$  in~\eqref{eq:forward_inv_set_phi}, we can use a nominal controller $\ub_{\ltt_k}^d = \Mb_{\ltt_k}^{-1}\nabla_{\z_k}\Phi_{\ltt_k}$. This choice meets the
% sufficient condition suggested by \cref{thm:clf_rmpflow_stability}; the CLF constraint in~\eqref{eq:qp} is automatically satisfied when $\zd_{\ltt_k}=0$.

% parameterizing $\f_{\ltt_k}$ as,
% \begin{equation}
%     \f_{\ltt_k} = -\nabla_{\z_k}\Phi_{\ltt_k} + \gamma_k(\|\zd_k\|)\,\hb_{\ltt_k},
% \end{equation}
% where $\gamma_{k}$ is a class $\KK$ function. The constraint~\eqref{eq:leaf_node_constraints} for $\hb_{\ltt_k}$ is hereby
% \begin{equation}\label{eq:leaf_node_constraints_h}
%     \begin{split}
%         \zd_k^\top\,\gamma_{k}(\|\zd_k\|)\,\h_{\ltt_k}\,&\leq\,-\zd_k^\top\bm\xi_{\Gb_{\ltt_k}} - \alpha_{k}(\|\zd_k\|).
%     \end{split}
% \end{equation}

% Then, we can replace the linear constraints in~\eqref{eq:qp} as~\eqref{eq:leaf_node_constraints_h} to guarantee that the system converges to the forward invariant set $C_\infty^\Phi$ defined in \eqref{eq:forward_inv_set_phi}.

Note that \flow can be interpreted as a soft version of the QP--CLF formulation \cite{morris2013sufficient} that enforces the decay-rates of \emph{all} Lyapunov function candidates~\eqref{eq:qp_clf}. Meanwhile, compared with the QP--CLF framework with slack variables~\cite{ames2014control} that requires the users to design the objective function trade off between control specifications, \flow provides a structured way to implicitly generate such an objective function so that the system is \emph{always} stable.
% \todo{In \cref{sec:connection}, we will further discuss the connections between the original \flow based on GDSs~\cite{cheng2018rmpflow} and the QP--CLF framework~\cite{morris2013sufficient,ames2014control}. }

It should be noted that the system can also be stabilized by directly enforcing a single constraint on the time derivative of the combined Lyapunov function candidate~\eqref{eq:lyap_decay_clf} at the root node, rather than enforcing the CLF constraint at every leaf node~\eqref{eq:lyap_decay_clf_leaf}. However, this can be less desirable: although the stability can be guaranteed for the resulting controller, the behavior of each individual subtask is no longer explicitly regulated. % when only a single CLF constraint is used at the root node.
By contrast, the approach with leaf-node CLF constraints allows the users to design and test the controllers from~\eqref{eq:qp} independently. This allows for designing controllers that can be applied to robots with different kinematic structures, which is a significant feature of \flow~\cite{cheng2018rmpflow}.

\section{Experimental Results}\label{sec:results}

In this section, we compare the proposed \flow--CLF framework with the original \flow framework~\cite{cheng2018rmpflow}. A video of the experimental results can be found at \url{https://youtu.be/eU_x8Yklv-4}. % %through both  simulation and robotic implementation
The original \flow framework~\cite{cheng2018rmpflow} is referred to as \flow--GDS to differentiate it from \flow--CLF.

\subsection{Simulation Results}

We present two simulated examples, a $2$-dimensional goal reaching task and a multi-robot goal reaching example. %the proposed framework can generate richer behaviors under different choices of nominal controllers % gds may not be the most efficient choice % injecting prior knowledge into the nominal controller may ``improve performance''

\subsubsection{2D Goal Reaching}

We first consider the 2D goal reaching task presented in~\cite{cheng2018rmpflow}. In this example, a planar robot with double-integrator dynamics is expected to move to a goal position without colliding into a circular obstacle. As is in~\cite{cheng2018rmpflow}, the \tree has a collision avoidance leaf-node RMP and a goal attractor leaf-node RMP. For the \flow--CLF framework, we use the collision avoidance RMP in~\cite{cheng2018rmpflow} and keep the choice of metrics and potential functions for the goal attractor RMP consistent with~\cite{cheng2018rmpflow}. For the goal attractor RMP, we present several nominal controllers: (i) a pure potential-based nominal controller $\fb_{pt}^d = \M\ub_{pt}^d = -\nabla\Phi$; (ii) a spiral nominal controller $\fb_{sp}^d = -\nabla\Phi + \|\zd\|\,\vb$, where $\vb$ is the potential-based controller rotated by $\pi/2$, i.e. $\vb = -R(\pi/2)\,\nabla\Phi$ with $R(\cdot)$ being the rotation matrix; and (iii) a sinusoidal controller $\fb_{sn}^d = -\nabla\Phi + \sin(t/4)\,\|\zd\|\,\vb$. For the minimally invasive controller, we use $\Pb=\I$ to minimize the Euclidean distance between the nominal controller and the solution to the optimization problem~\eqref{eq:qp}. We implement the \flow--GDS framework with the same choice of parameters as~\cite{cheng2018rmpflow}.
The trajectories under different nominal controllers are shown in \cref{fig:attractor_2d}. Although it is possible that similar behaviors can be realized with the \flow--GDS framework with a careful redesign of the metric and potential function, the \flow--CLF framework can produce a rich class of behaviors without being concerned with the geometric properties of the subtask manifold.

\begin{figure}
    \centering
    \vspace{2mm}
    \begin{subfigure}{.61\columnwidth}
  \centering
  \resizebox{!}{1.2in}{\includegraphics{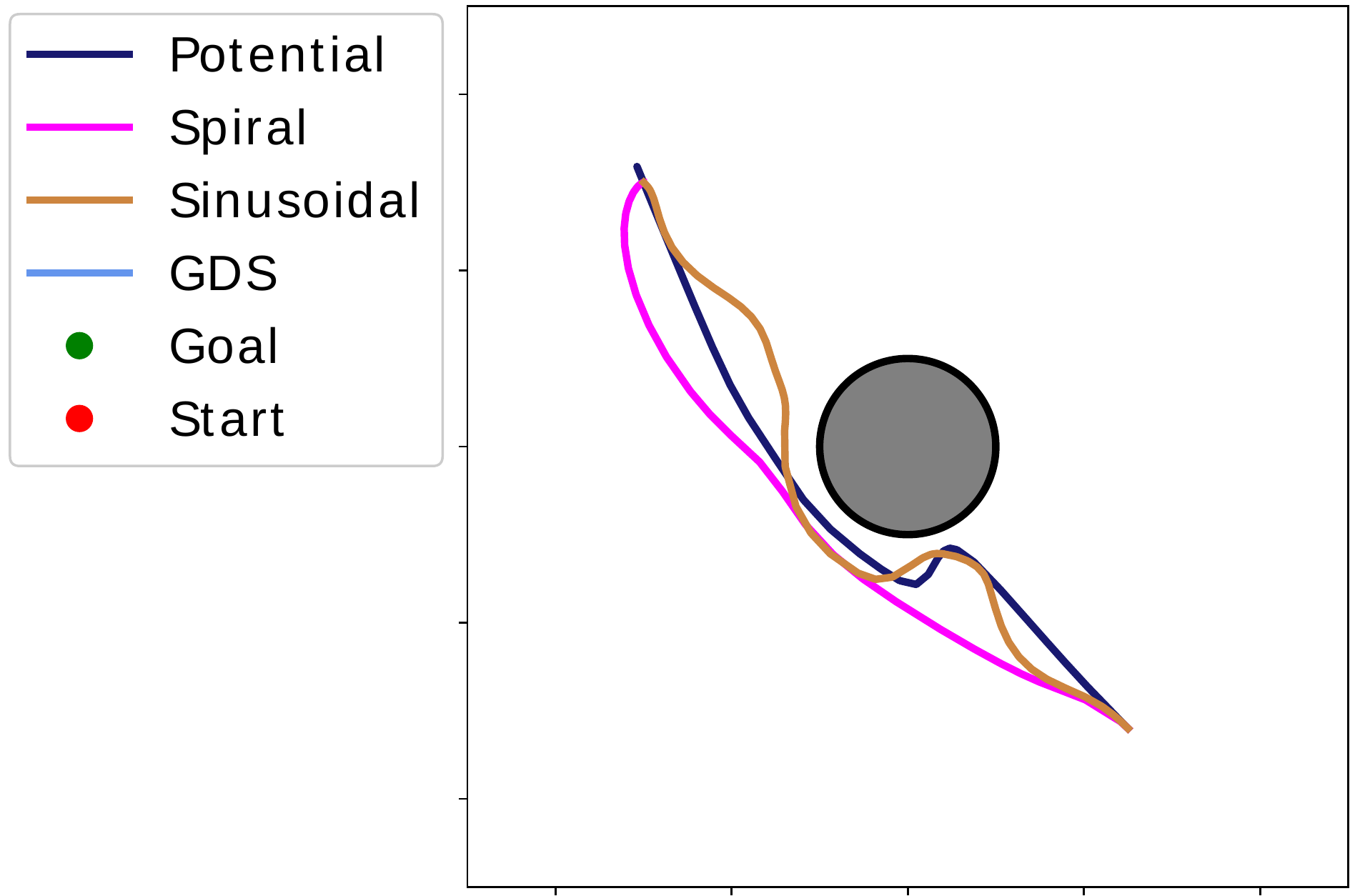}}
  \caption{\flow--CLF}
  \label{fig:attractor_2d_clf}
\end{subfigure}%
\begin{subfigure}{.39\columnwidth}
    \centering
  \resizebox{!}{1.2in}{\includegraphics{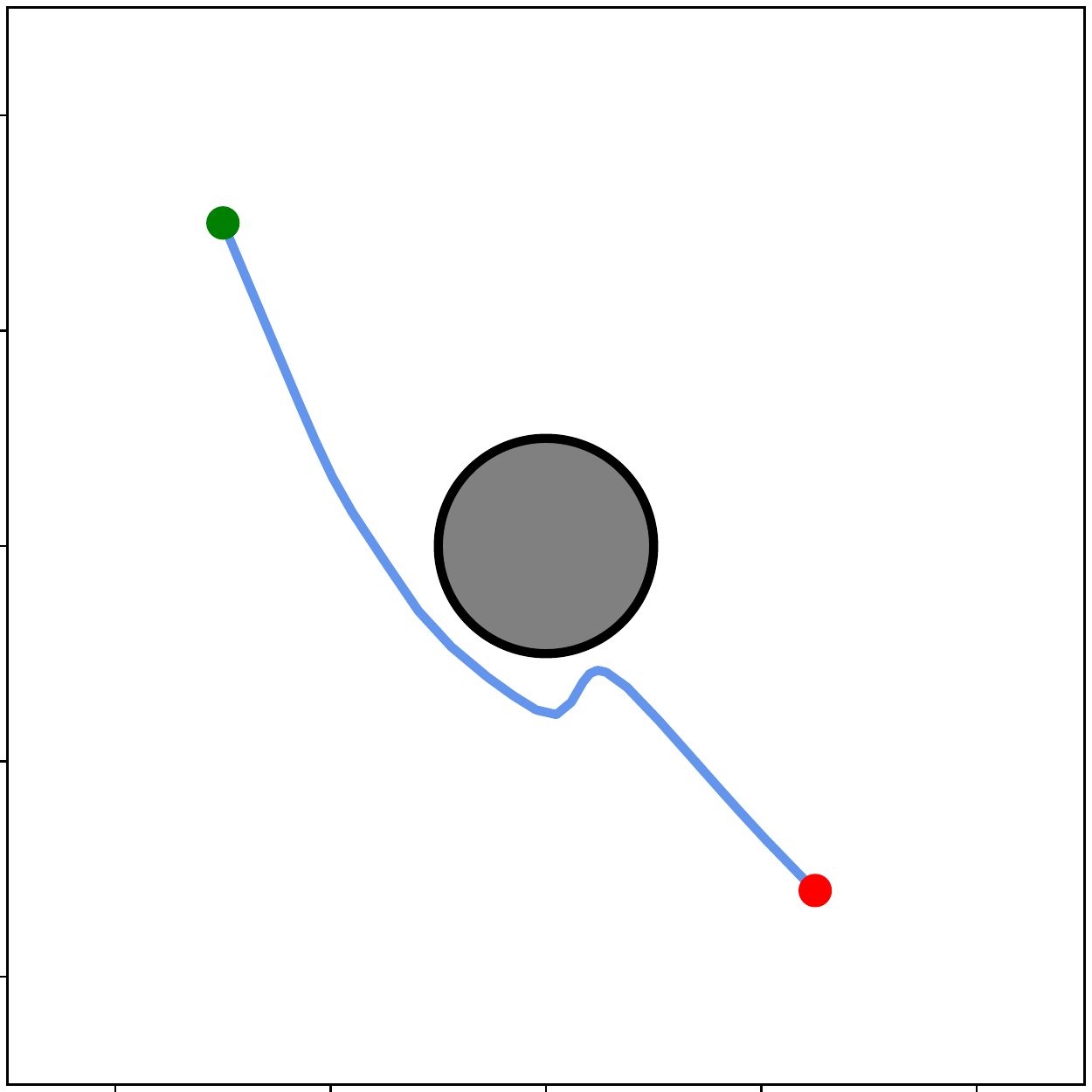}}
  \caption{\flow--GDS~\cite{cheng2018rmpflow}}
  \label{fig:attractor_2d_gds}
\end{subfigure}
	\caption{
	2D goal reaching task with a circular obstacle (grey). (a) \flow--CLF with three choices of nominal controllers, resulting in different goal reaching behaviors. (b) \flow--GDS with the goal attractor given by a GDS. The behavior is limited by the choice of the metric and the potential function. \vspace{-10mm}
	}
    \label{fig:attractor_2d}
\end{figure}

\subsubsection{Multi-Robot Goal Reaching}

% The nominal controller can be viewed as a way to inject design knowledge for subtasks into the resulting controllers.
 \flow--CLF guarantees system stability even when the nominal controllers are not inherently stable or asymptotically stable. Therefore, the user can incorporate design knowledge given by, e.g. motion planners, human demonstrations or even heuristics, into the nominal controllers. To illustrate this, we consider a multi-robot goal reaching task, where the robots are tasked with moving to the opposite side of the arena without colliding. If the robots move in straight lines, their trajectories would intersect near the center of the arena. Due to the symmetric configuration, the system can easily deadlock with robots moving very slowly or stopping near the center to avoid collisions. This problem can be fixed if the symmetry is broken. One possible solution is to design nominal controllers for the goal attractors so that the robots move along curves.

We compare the spiral goal attractor RMP with the GDS goal attractor RMP presented in~\cite{li2019multi}. In both cases, an \tree structure similar to the centralized \tree structure in~\cite{li2019multi} is used. We define collision avoidance for pairs of robots in the \tree with the same choice of parameters. The trajectories of the robots under the spiral nominal controllers are shown in~\cref{fig:multi_robot_clf}. The spiral controllers produce smooth motion, whereas the GDS goal attractors produce jerky motion when the robots are near the center due to deadlock caused by the symmetric configuration (\cref{fig:multi_robot_gds}).

\begin{figure}
    \centering
    \vspace{2mm}
    \begin{subfigure}{.61\columnwidth}
  \centering
  \resizebox{!}{1.2in}{\includegraphics{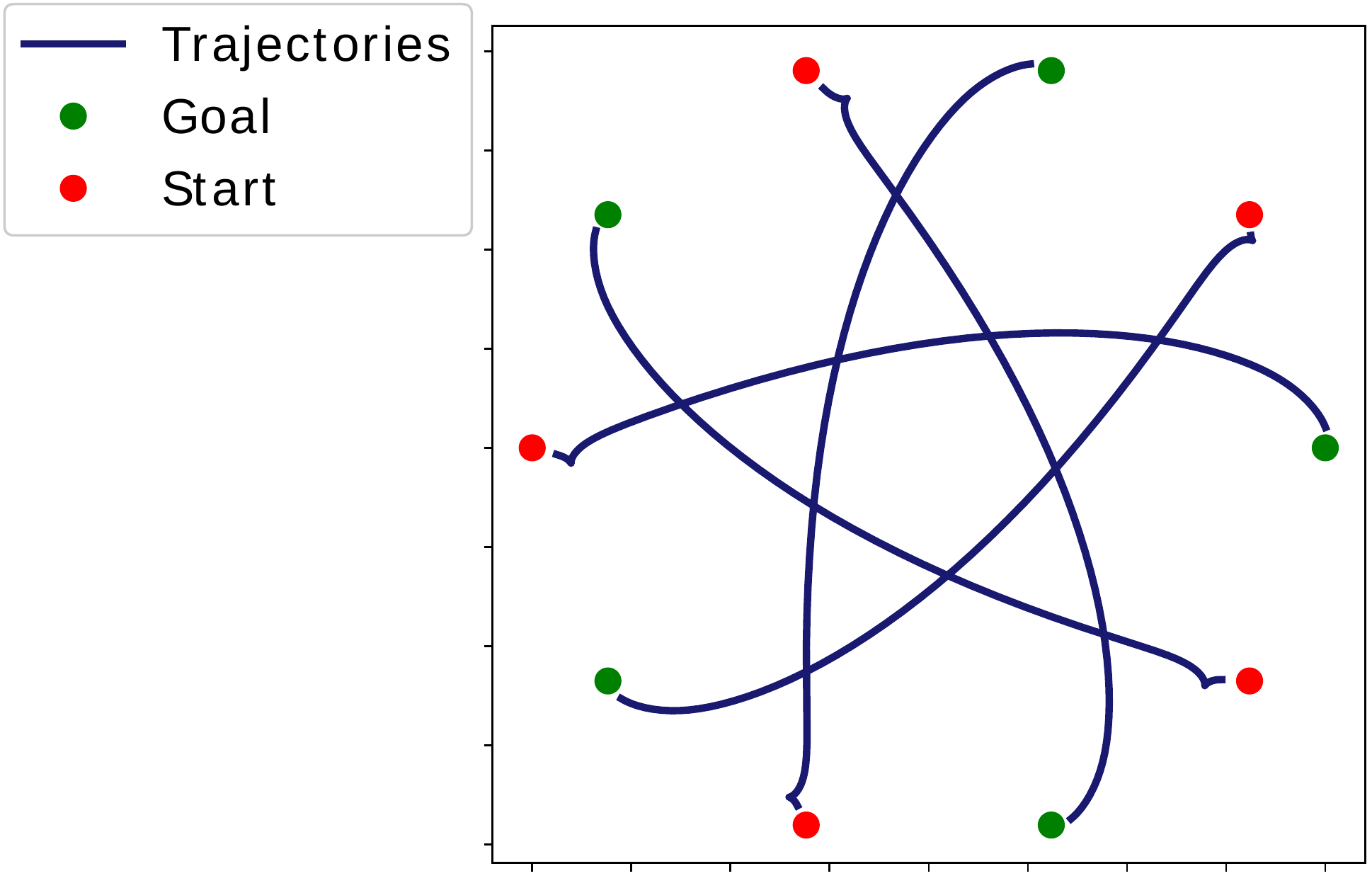}}
  \caption{\flow--CLF}
  \label{fig:multi_robot_clf}
\end{subfigure}%
\begin{subfigure}{.39\columnwidth}
  \centering
  \resizebox{!}{1.2in}{\includegraphics{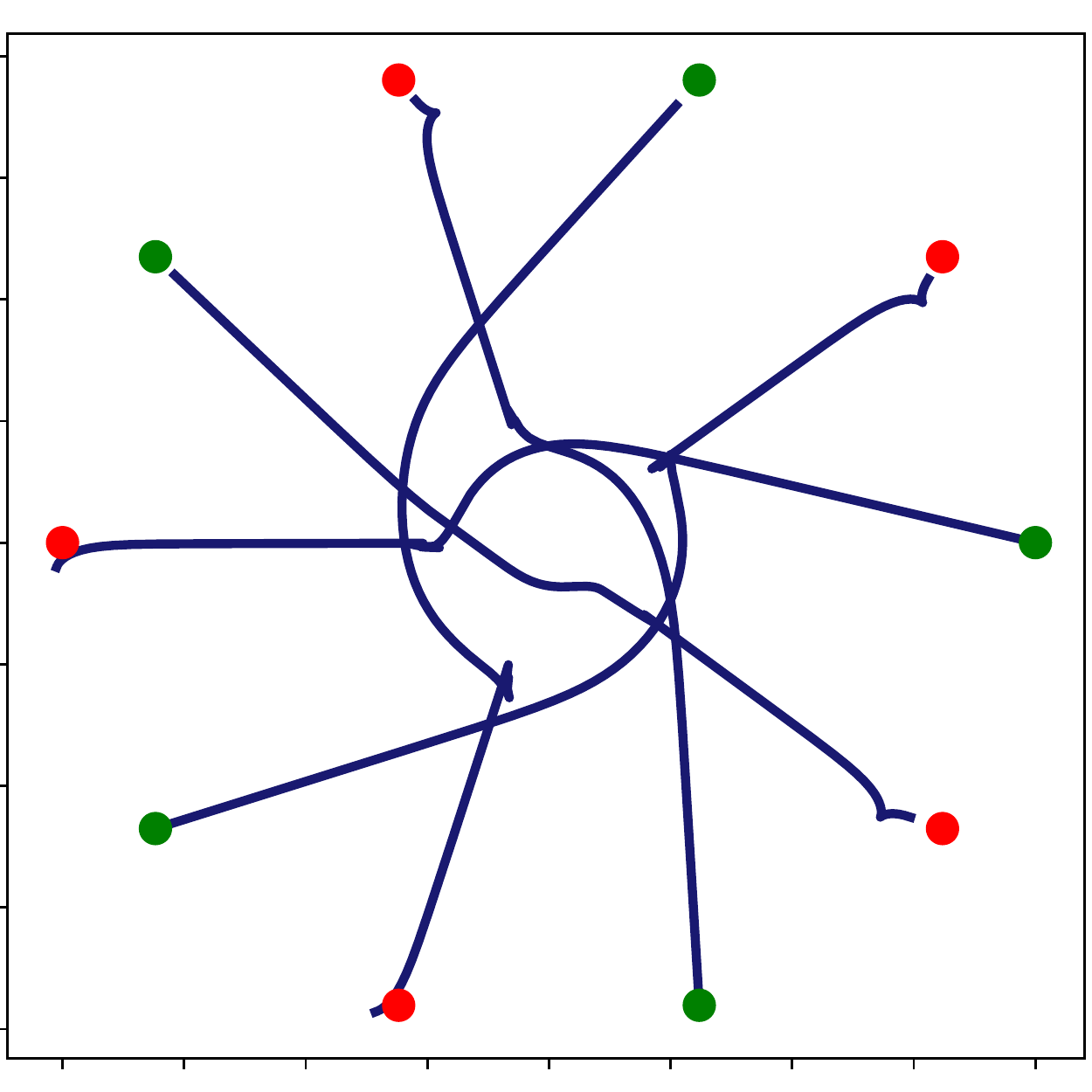}}
  \caption{\flow--GDS~\cite{cheng2018rmpflow}}
  \label{fig:multi_robot_gds}
\end{subfigure}
	\caption{
	Multi-robot goal reaching task. (a) \flow--CLF with spiral nominal controllers. The robots move to their goal smoothly. (b) \flow--GDS with the goal attractor given by a GDS. Due to the symmetry of the configuration, the system suffers from deadlock when the robots are near the center: the robots oscillate around the deadlock configuration.\vspace{-4mm}
	}
    \label{fig:multi_robot}
\end{figure}

\subsection{Robotic Implementation}

We present an experiment conducted on the Robotarium~\cite{pickem2017robotarium}, a remotely accessible swarm robotics platform. In the experiment,  five robots are tasked with preserving a regular pentagon formation while the leader has an additional task of reaching a goal. We use the same \tree structure and parameters for most leaf-node RMPs as described in the formation preservation experiment  in~\cite{li2019multi}. The only difference is that we replace the GDS goal attractor in~\cite{li2019multi} with the spiral nominal controller augmented with the CLF condition~\eqref{eq:qp}.
\cref{fig:exp} presents the snapshots from the formation preservation experiment. In \cref{fig:exp_clf_0}--\cref{fig:exp_clf_38}, we see that the leader  approaches the goal with a spiral trajectory specified by the nominal controller, while other subtask controllers preserve distances and avoid collision. This shows the efficacy of our controller synthesis framework. By contrast, the robot moves in straight lines under the goal attractor given by the GDS (see \cref{fig:exp_gds_0}-\cref{fig:exp_gds_19}). Although it could be possible to redesign the subtask manifold such that there exists a GDS that produces similar behaviors, the \flow--CLF framework provides the user additional flexibility to shape the behaviors without worrying about the geometric properties of the subtask manifolds.

\begin{figure*}
    \centering
    \vspace{2mm}
\begin{subfigure}{.32\textwidth}
  \centering
  \resizebox{!}{1.1in}{\includegraphics{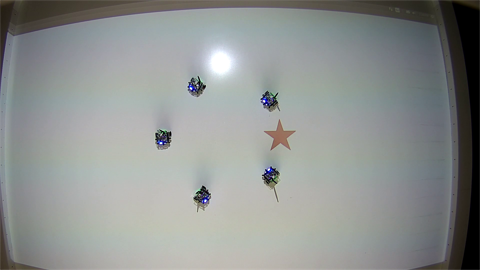}}
  \caption{Spiral CLF controller: $t=0s$}
  \label{fig:exp_clf_0}
\end{subfigure}
\begin{subfigure}{.32\textwidth}
  \centering
  \resizebox{!}{1.1in}{\includegraphics{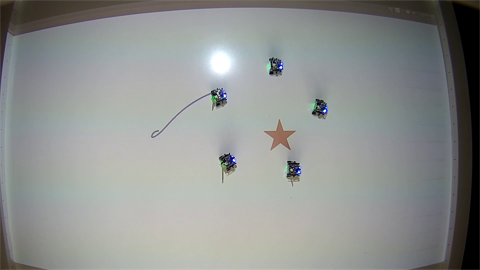}}
  \caption{Spiral CLF controller: $t=18s$}
  \label{fig:exp_clf_18}
\end{subfigure}
\begin{subfigure}{.32\textwidth}
  \centering
  \resizebox{!}{1.1in}{\includegraphics{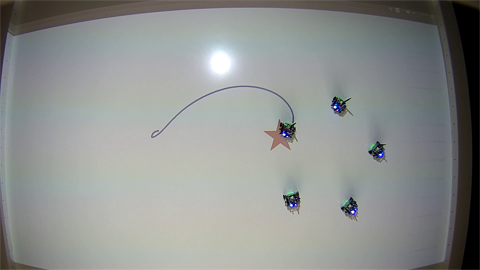}}
  \caption{Spiral CLF controller: $t=38s$}
  \label{fig:exp_clf_38}
\end{subfigure}
\\
\vspace{3mm}
\begin{subfigure}{.32\textwidth}
  \centering
  \resizebox{!}{1.1in}{\includegraphics{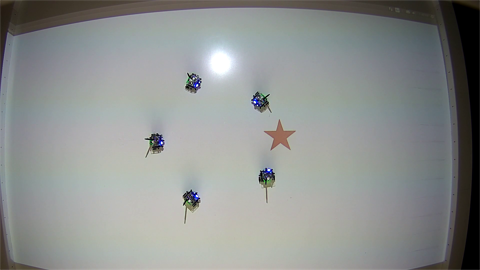}}
  \caption{GDS controller: $t=0s$}
  \label{fig:exp_gds_0}
\end{subfigure}
\begin{subfigure}{.32\textwidth}
  \centering
  \resizebox{!}{1.1in}{\includegraphics{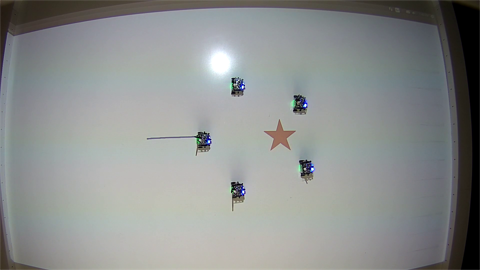}}
  \caption{GDS controller: $t=8s$}
  \label{fig:exp_gds_8}
\end{subfigure}
\begin{subfigure}{.32\textwidth}
  \centering
  \resizebox{!}{1.1in}{\includegraphics{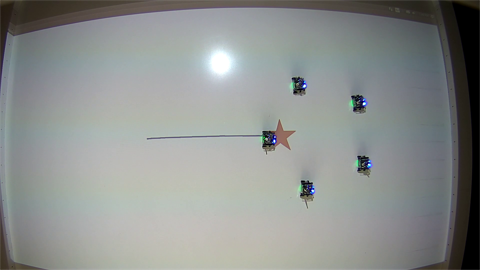}}
  \caption{GDS controller: $t=19s$}
  \label{fig:exp_gds_19}
\end{subfigure}

	\caption{
	Multi-robot formation preservation task. The robots are tasked with preserving a regular pentagon formation while the leader has an additional task of reaching a goal position. (a) \flow--CLF with a spiral nominal controller. (b) \flow--GDS. The goal (red star) and the trajectories (blue curves) of the leader robot are projected onto the environment through an overhead projector. \flow--CLF shapes the goal-reaching behavior through a spiral nominal controller. \vspace{-2mm}
	}
    \label{fig:exp}
\end{figure*}

\section{Conclusions}\label{sec:conclusions}
We consider robot control with multiple control specifications by adopting Riemannian Motion Policies (RMPs), a recent concept in robotics for describing control policies on manifolds, and \flow, the computational structure to combine these controllers. The stability results of \flow is re-established and extended through a rigorous CLF treatment. This new analysis suggests that any subtask controllers satisfying certain CLF constraints can be stably combined by \flow, while the original analysis given in~\cite{cheng2018rmpflow} only provides stability guarantees for a limited type of controller.
%, and justifies the application of \flow to this class of subtask controllers.
Based on this analysis, we propose a new computational framework, \flow--CLF, to stably combine individually designed subtask controllers.
This formulation provides users the flexibility of shaping behaviors of subtasks through nominal subtask controllers given by, e.g. heuristics, human demonstrations, and motion planners. The proposed \flow--CLF framework is validated through numerical simulation and deployment on real robots.

\bibliographystyle{ieeetr}
\bibliography{references}

% \clearpage
% \input{appendix.tex}

\end{document}